\documentclass[a4paper,UKenglish,cleveref, autoref]{article}
\usepackage[l2tabu, orthodox]{nag}
\usepackage[utf8]{inputenc}
\usepackage[pdfusetitle]{hyperref}
\usepackage{microtype}
\usepackage{amsmath,amssymb,amsthm}
\newtheorem{definition}{Definition}
\newtheorem{proposition}{Proposition}
\newtheorem{lemma}{Lemma}
\newtheorem{theorem}{Theorem}

\newtheorem{corollary}{Corollary}

\usepackage{fullpage}

\usepackage{enumerate} %Gives ability, e.g., \begin{enumerate}[Case 1.]

\usepackage{xspace}

\newcommand{\N}{\ensuremath{\mathbb N}\xspace}

\newcommand{\F}{\ensuremath{\mathbb F}\xspace}

\newcommand{\cI}{\ensuremath{\mathcal I}\xspace}

\newcommand{\cP}{\ensuremath{\mathcal P}\xspace}

\newcommand{\cT}{\ensuremath{\mathcal T}\xspace}
\newcommand{\cY}{\ensuremath{\mathcal Y}\xspace}

\usepackage{todonotes}

\newcommand{\cpc}{\mathrm{cap}}

\title{Quasipolynomial multicut-mimicking networks
  and kernels for multiway cut problems\footnote{A preliminary version
    of this paper with weaker results was presented at ICALP 2020~\cite{selfICALP}.
    Compared to that version, the current manuscript has a much simpler
    correctness argument for the marking procedure, as well as adding
    the constructive version of the result.}}
\author{Magnus
  Wahlstr\"om\footnote{\texttt{Magnus.Wahlstrom@rhul.ac.uk};
    Department of Computer Science, Royal Holloway, University of
    London, UK.}}

\begin{document}

\maketitle
\begin{abstract}
  We show the existence of an exact mimicking network of $k^{O(\log k)}$ edges
  for minimum multicuts over a set of terminals in an undirected graph,
  where $k$ is the total capacity of the terminals. 
  Furthermore, using the best available approximation algorithm for 
  \textsc{Small Set Expansion}, we show that a mimicking network of
  $k^{O(\log^3 k)}$ edges can be computed in randomized polynomial time.
  As a consequence, we show quasipolynomial kernels for several problems,
  including \textsc{Edge Multiway Cut},
  \textsc{Group Feedback Edge Set} for an arbitrary group,
  % \textsc{0-Extension} for integer-weighted metrics,
  and \textsc{Edge Multicut} parameterized by the solution size and the number
  of cut requests. 
  The result combines the matroid-based irrelevant
  edge approach used in the kernel for \textsc{$s$-Multiway Cut}
  with a recursive decomposition and sparsification of the graph along
  sparse cuts.  
  This is the first progress on the kernelization of \textsc{Multiway Cut} problems 
  since the kernel for \textsc{$s$-Multiway Cut} for constant value of $s$ (Kratsch and Wahlstr\"om, FOCS 2012). 
\end{abstract}

%%%%%%%%%%%%%%%%%%%%%%%%%%%%%%%%%%%%%%%%%%%%%%%%%%%%%%

\section{Introduction}

Graph separation questions are home to some of the most intriguing
open questions in theoretical computer science.  In approximation
algorithms, the well-known \emph{unique games conjecture} (UGC)
has been central to the area for close to two decades,
and is closely related to graph separation problems.  
Even more directly, the \emph{small set expansion hypothesis},
proposed by Raghavendra and Steurer~\cite{RaghavendraS10STOC},
roughly states that it is NP-hard to approximate the \textsc{Small Set Expansion}
problem (SSE) up to a constant factor, where SSE is the problem of
finding a small-sized set in a graph with minimum expansion.
(More precise statements are given in Section~\ref{sec:graphsep}.)
For the general case, despite significant research, the best
polynomial-time result is an $O(\log n)$-approximation due to R\"acke~\cite{Racke08decomp},
but stronger results are known for special cases.
In particular, if the size bound on the set is $|S|\leq s$,
then Bansal et al.~\cite{BansalFKMNNS14SICOMP} show an algorithm with
an approximation ratio of $O(\log n/\sqrt{\log s})$. 

Another interesting notion from parameterized complexity is
\emph{kernelization}. Informally, a kernelization algorithm
is a procedure that takes an input of a parameterized, usually NP-hard problem
and \emph{reduces} it to an equivalent instance of size bounded in the parameter,
e.g., by discarding irrelevant parts of the input or transforming
some part of the input into a smaller object with equivalent behaviour.
For example, the seminal Nemhauser-Trotter theorem on the
half-integrality of \textsc{Vertex Cover}~\cite{NemhauserT75} implies
that an instance of \textsc{Vertex Cover} can be reduced
to have at most $2k$ vertices, where $k$ is the bound given on the solution size.
On the flip side, Fortnow and Santhanam~\cite{FortnowS11}
and Bodlaender et al.~\cite{BodlaenderDFH09} gave a framework to exclude
the existence of a kernel of any polynomial size, under a standard complexity-theoretic 
conjecture. An extensive collection of upper and lower bounds
for kernelization exists (see, e.g., the recent book of Fomin et
al.~\cite{FominLSZkernelsbook}), but a handful of central
``hard questions'' remain unanswered. One of the most 
notorious is \textsc{Multiway Cut}.

Let $G=(V,E)$ be a graph and $T \subseteq V$ a set of terminals in $G$.
An \emph{(edge) multiway cut} for $T$ in $G$ is a set of edges
$X \subseteq E$ such that no two terminals are connected in $G-X$,
and \textsc{Multiway Cut} is the problem of finding a multiway cut
of at most $k$ edges, given a parameter $k$.  The problem is FPT~\cite{Marx06MWC}
and NP-hard for $|T| \geq 3$~\cite{DahlhausJPSY94}.
Using methods from matroid theory, 
Kratsch and Wahlstr\"om~\cite{KratschW20JACM} were able to show 
that if $|T| \leq s$, then \textsc{Multiway Cut} has a kernel with
$O(k^{s+1})$ vertices, hence the problem has a polynomial kernel for
every constant $s$.  However, if $|T|$ is unbounded, the only known size bound for a
kernel is $2^{O(k)}$, following from the FPT algorithm~\cite{Marx06MWC},
and the question of whether \textsc{Multiway Cut} has a polynomial kernel in the
general case is completely open.

We make progress on this question by showing that \textsc{Multiway Cut}
and several related problems have \emph{quasipolynomial} kernels,
i.e., kernels of size $k^{\log^{O(1)} k}$.  Furthermore, the degree in
the exponent depends on the best available approximation algorithm for
\textsc{Small Set Expansion}.  With the current state of the art, we
are able to show kernels of size $k^{O(\log^3 k)}$; and if the small
set expansion problem has a constant-factor approximation,
the result would be kernels of size $k^{O(\log k)}$. 

The result goes via showing the existence of a kind of
\emph{mimicking network} for the problem; or more generally,
a network of quasipolynomial size mimicking the behaviour of $(G,T)$
for all multicut instances over $T$.  We review these notions next.

\subsection{Mimicking networks  and multiway cut sparsifiers}

Although kernelization is most commonly described in terms of
polynomial-time preprocessing as above, there is also a clear 
connection with \emph{succinct information representation}.
For example, consider a graph $G=(V,E)$ with a set of $k$ terminals
$T \subseteq V$. The pair $(G,T)$ is referred to as a \emph{terminal network}. 
A \emph{mimicking network} for $(G,T)$
is a graph $G'=(V',E')$ with $T \subseteq V'$ such that
for any sets $A, B \subseteq T$, the min-cut between $A$ and $B$
in $G$ and $G'$ have the same value.
A mimicking network of size bounded in $k$ always exists,
but the size of $G'$ can be significant. 
The best known general upper bound is double-exponential in~$k$~\cite{HagerupKNR98JCSS,KhanR14Mimic}, and
there is an exponential lower bound~\cite{KrauthgamerR13SODA}.
Better bounds are known for special graph classes, but even for planar graphs 
the best possible general bound has $2^{\Theta(k)}$ vertices~\cite{KrauthgamerR13SODA,KarpovPZ19Mimic} 
(see also recent improvements by Krauthgamer and Rika~\cite{KrauthgamerR20SIDMA}).

A related notion is \emph{cut sparsifiers}, which solve the same task
up to some approximation factor $q \geq 1$~\cite{Moitra09FOCS,LeightonMoitra10STOC},
typically $q=\omega(1)$ in the general case.  We focus on mimicking
networks; see Krauthgamer and Rika~\cite{KrauthgamerR20SIDMA} for an overview of cut sparsifiers.
However, we note that for general graphs, constant-factor cut sparsifiers
are not known with any size bound smaller than the size guarantee for an exact mimicking network.

%[7] question and $2^{2^k}$ HagerupKNR98JCSS

%[9] improved KhanR14Mimic

%[10] exp. l.b. KrauthgamerR13SODA

%[10] exp. u.b. for planar KrauthgamerR13SODA

%[8] exp l.b. planar KarpovPZ19Mimic

%[6]. u.b. single-face-planar -

%[3] twd -

%and

%sparsif. = approx 13,11 and 5,1
%13, 11 = moitra, leighton/moitra \cite{Moitra09FOCS,LeightonMoitra10STOC}
%5 = new from old \cite{EnglertGKRTT14SICOMP}
%1 = towards (AGK14) \cite{AndoniGK14SODA}

However, if we include the capacity of the set of terminals in the
bound (and if edges have integer capacity), then significantly
stronger results are possible. Chuzhoy~\cite{Chuzhoy12STOC} showed
that if the total capacity of $T$ is $\cpc_G(T)=\sum_{t \in T} d(t)=k$, 
then there exists an $O(1)$-approximate cut sparsifier of size $O(k^3)$.
Kratsch and Wahlstr\"om~\cite{KratschW20JACM} sharpened this to
an exact mimicking network with $O(k^3)$ edges, which furthermore can be computed in
randomized polynomial time.  
This is particularly remarkable given that the
network has to replicate the exact cut-value for exponentially many
pairs $(A,B)$.
The network can be constructed via contractions in $G$.\footnote{The results
of \cite{KratschW20JACM} are phrased in terms of vertex cuts, but the above follows easily from \cite{KratschW20JACM}.}
This built on an earlier result that used linear representations of matroids to
encode the sizes of all $(A,B)$-min cuts into an object
using $\tilde O(k^3)$ bits of space~\cite{KratschW14TALG},
although this earlier version did not produce an explicit graph, i.e., not a mimicking
network.  

These results had significant consequences for kernelization.
The succinct representation in~\cite{KratschW14TALG} was 
used to produce a (randomized) polynomial kernel for
the \textsc{Odd Cycle Transversal} problem, thereby solving a notorious
open problem in parameterized complexity~\cite{KratschW14TALG};
and the mimicking network of~\cite{KratschW20JACM} brought further (randomized) polynomial kernels
for a range of problems, in particular including \textsc{Almost 2-SAT},
i.e., the problem of satisfying all but at most $k$ clauses of a given
2-CNF formula.

Similar methods are relevant for the question of separating a set of
terminals into more than two parts. Let $(G,T)$ be a terminal network,
and let $\cT=T_1 \cup \ldots \cup T_s$ be a partition of $T$. 
A \emph{multiway cut for $\cT$} is a set of edges $X \subseteq E(G)$
such that $G-X$ contains no path between any pair of terminals
$t \in T_i$ and $t' \in T_j$, $i \neq j$. 
Let us define a \emph{multicut-mimicking network} for $(G,T)$ as a
terminal network $(G', T)$ where $T \subseteq V(G')$ and for every
partition $\cT=T_1 \cup \ldots \cup T_s$ of $T$, 
the size of a minimum multiway cut for $\cT$ is identical in $G$ and $G'$.
(The term \emph{multicut-mimicking}, as opposed to \emph{multiway cut-mimicking},
is justified; see Section~\ref{sec:multi-mimicking}.)
The minimum size of a multicut-mimicking network, in terms of $k=\cpc_G(T)$,
appears to lie at the core of the difficulty of the question of a
polynomial kernelization of \textsc{Multiway Cut}.
The kernel for \textsc{$s$-Multiway Cut} mentioned above builds
on the computation of a mimicking network of size $O(k^{s+1})$ for partitions of $T$
into at most $s$ parts~\cite{KratschW20JACM}. The kernel for \textsc{$s$-Multiway Cut}
then essentially follows from considering the partition $\cT=\{t_1\} \cup \ldots \cup
\{t_s\}$ of a set $T$ of $|T|=s$ terminals, along with known reduction rules bounding
$\cpc_G(T)$. We are not aware of any non-trivial lower bounds on
the size of a multicut-mimicking network in terms of $k$; it seems completely
consistent with known bounds that every terminal network $(G,T)$ would
have a multicut-mimicking network of size poly$(k)$, even for
partitions into an unbounded number of sets.

In this paper, we show that any terminal network $(G,T)$ with $\cpc_G(T)=k$
admits a multicut-mimicking network $(G',T)$ where $|V(G')|=k^{O(\log k)}$;
and furthermore, a network with $|V(G')| = k^{O(\log^{O(1)} k)}$ can be
computed in randomized polynomial time, using a
sufficiently good approximation algorithm for a graph separation
problem similar to \textsc{Small Set Expansion} (SSE).
We also see a tradeoff between the quality of the approximation algorithm
and the size of $(G',T)$.
Using the algorithm of Bansal et
al.~\cite{BansalFKMNNS14SICOMP}, we achieve $|V(G')| = k^{O(\log^3 k)}$;
and if the small set expansion hypothesis were false and SSE
had a constant-factor approximation algorithm,
then the bound $|V(G')|=k^{O(\log k)}$ would be achievable in polynomial time.
%has an approximation algorithm with a ratio of
%$
%  \alpha(n,k) = \log^{1-\varepsilon} n \cdot \log^{O(1)} k
%$
%for some $\varepsilon > 0$, where $k$ is the number of edges cut in the optimal solution,
%then $(G',T)$ can be computed efficiently, with $|V(G')|$ being quasipolynomial in $k$.
%Such an algorithm goes beyond the bounds of what is currently known --
%namely, a ratio of $O(\log n)$ due to R\"acke~\cite{Racke08decomp},
%improved for certain regimes by Bansal et al.~\cite{BansalFKMNNS14SICOMP} --
%but does not appear to be excluded by any established hardness conjecture.
We leave open the questions of whether there always exists a multicut-mimicking network
of size $k^{O(1)}$, as well as the question of whether a network of size $k^{O(\log k)}$
can be computed through means other than a constant-factor approximation for SSE. 
%We also consider the existence result very interesting in its own right,
%and invite further study of capacity-based bounds for
%multicut-mimicking networks; in particular, whether a poly$(k)$-sized
%multicut-mimicking network always exists.
%The results strongly suggest the existence of a quasipolynomial kernel for \textsc{Edge Multiway Cut}.
%We leave open the question of existence of a poly$(k)$-sized multicut-mimicking network in general.

As a side note, we note that a 2-approximate ``multicut sparsifier'' of size poly$(k)$
can be computed efficiently using known methods.  Specifically, we observe that the
above-mentioned mimicking network of $O(k^3)$ edges for standard terminal cuts
in a terminal network $(G,T)$ with $\cpc(T)=k$~\cite{KratschW20JACM}
implies a terminal network $(G',T)$ where $\cpc(T)=k$,
$|E(G')|=O(k^3)$, and where for every partition $\cT$
the costs of minimum multiway cuts for $\cT$ in $(G',T)$ and $(G,T)$
differ by at most a factor of 2 (see Lemma~\ref{cut-cover-is-apxmcc}).

\emph{Flow sparsifiers.} Finally, similarly to cut sparsifiers, there
is a notion of a \emph{flow sparsifier} of a terminal network $(G,T)$.
Here the goal is to approximately preserve the minimum congestion for
any multicommodity flow on $(G,T)$. Chuzhoy~\cite{Chuzhoy12STOC} showed
flow sparsifiers with quality $O(1)$ and with $k^{O(\log \log k)}$ vertices,
where $k$ is the total terminal capacity; for further results on
achievable bounds for flow sparsifiers, see~\cite{AndoniGK14SODA,EnglertGKRTT14SICOMP}.
However, the notion is incomparable to multicut-mimicking networks,
because even an exact flow sparsifier would be
subject to the corresponding multicommodity flow-multicut approximation gap,
which is $\Theta(\log k)$ in the worst case~\cite{GargVY96}.

\emph{Further related work.} The general approach of decomposing a
graph along sparse cuts is well established; cf.~R\"acke~\cite{Racke02congestion}
and follow-up work. For further applications of matroid tools to
kernelization, see Hols and Kratsch~\cite{HolsK18SFVS}, Kratsch~\cite{Kratsch18VClovplum},
and Reidl and Wahlstr\"om~\cite{ReidlW18ICALP}.

\subsection{Our results}

We show the following.

%We show how the existence of an approximation algorithm for a certain
%cut-problem, slightly better than the current state of the art, would
%imply a quasipolynomial kernel for the \textsc{Edge Multiway Cut}
%problem. We also show the existence of a multicut-covering set of
%$k^{O(\log k)}$ edges. 
%
%Formally, we show the following.  

\begin{theorem}
  \label{thm:intro:main}
  Let $A$ be an approximation algorithm for \textsc{Small Set Expansion}
  with an approximation ratio of $\alpha(n,k)$, %=O(\log^{1-\varepsilon} n \log^d k)$,
  where %$\varepsilon > 0$, $d=O(1)$, and
  $k$ is the number of edges cut in the optimal solution.
  Let $(G,T)$ be a terminal network with $\cpc_G(T)=k$. 
  Then there is a set $Z \subseteq E(G)$ with $|Z| = k^{O(\alpha(n,k) \log k)}$
  such that for every partition $\cT=T_1 \cup \ldots \cup T_s$ of $T$, 
  there is a minimum multiway cut $X$ for $\cT$ such that $X \subseteq Z$.
  Furthermore, $Z$ can be computed in randomized polynomial time using  calls to $A$. 
\end{theorem}

%The precise requirement for the approximation algorithm is slightly relaxed from the above.
%We refer to the precise algorithm we need as a \emph{sublogarithmic terminal expansion
%  tester}; see~\autoref{def:set-finder}.
Unfortunately, the best known ratio for SSE which can be stated purely in terms of $k$ and $n$
is $\alpha(n,k)=O(\log n)$~\cite{Racke08decomp}, and plugging this into the above formula
yields a vacuous result.  However, by a more careful analysis we are able to
show a better bound using the algorithm of Bansal et al.~\cite{BansalFKMNNS14SICOMP}.
In summary, we get the following.

\begin{corollary}
  \label{cor:main} Let $(G,T)$ be a terminal network with $\cpc_G(T)=k$. The following holds.
  \begin{enumerate}
  \item There is a multicut-mimicking network for $(G,T)$ with $k^{O(\log k)}$ edges.
  \item A multicut-mimicking network for $(G,T)$ with $k^{O(\log^3 k)}$ edges can be computed in randomized polynomial time.
  \end{enumerate}
\end{corollary}

The latter implies several breakthrough results in kernelization, as follows.
We refer to previous kernelization work~\cite{KratschW20JACM} for the necessary definitions. 

\begin{corollary}
  The following problems have randomized quasipolynomial kernels. 
  \begin{enumerate}
  \item \textsc{Edge Multiway Cut} parameterized by solution size.
  \item \textsc{Edge Multicut} parameterized by the solution size and the number of cut requests.
  \item \textsc{Group Feedback Edge Set} parameterized by solution size, for any group.
  \item \textsc{Subset Feedback Edge Set} with undeletable edges, parameterized by
    solution size.
  %\item \textsc{0-Extension} for integer-weighted graphs, parameterized by solution cost.
  \end{enumerate}
\end{corollary}

\section{Preliminaries}

A \emph{parameterized problem} is a decision problem where inputs
are given as pairs $(X,k)$, where $k$ is the \emph{parameter}.
A \emph{polynomial kernelization} is a polynomial-time procedure 
that maps an instance $(X,k)$ to an instance $(X',k')$
where $(X,k)$ is positive if and only if $(X',k')$ is positive,
and $|X'|, k' \leq g(k)$ for some function $g(k)$ referred to as the \emph{size} of the
kernel.  A problem has a \emph{polynomial kernel} if it has a kernel where $g(k)=k^{O(1)}$.
We extend this to discuss \emph{quasipolynomial kernels}, which is the case
that $g(k)=k^{\log^{O(1)} k}$.  For more on parameterized complexity
and kernelization, see~\cite{CyganFKLMPPS15PCbook,FominLSZkernelsbook}.

For a graph $G=(V,E)$ and sets $A, B \subseteq V$,
we let $E_G(A,B)=\{uv \in E \mid u \in A, v \in B\}$.
As shorthand for $S \subseteq V$ we also write $E(S)=E(S,S)$,
$\partial_G(S)=E_G(S,V \setminus S)$,
and $\delta_G(S)=|\partial_G(S)|$. 
The \emph{total capacity} of a set of vertices $S$ in a graph $G$ is
\[
\cpc_G(S) := \sum_{v \in S} d(v).
\]
In all cases, we may omit the index $G$ if understood from context.

%The \emph{line graph} $L(G)$ of a graph $G=(V,E)$ is
%the graph on vertex set $V(L(G))=\{v_e \mid e \in E(G)\}$
%where two vertices $v_e, v_f \in V(L(G))$ are adjacent
%if and only if $e \cap f \neq \emptyset$. 

\subsection{Multicut-mimicking networks}
\label{sec:multi-mimicking}

Let $G=(V,E)$ be a graph and $T \subseteq V$ a set of terminals
with $\cpc(T)=k$.  An \emph{edge multiway cut} for $T$ in $G$
is a set of edges $X \subseteq E$ such that no two vertices in $T$
are connected in $G-X$.  More generally, let $\cT=\{T_1,\ldots,T_r\}$ be a partition of $T$.
Then an \emph{edge multiway cut for $\cT$ in $G$} is a set of edges $X \subseteq E$
such that in $G-X$ every connected component contains terminals from
at most one part of $\cT$.  Hence a multiway cut for $(G,T)$ is
equivalent to a multiway cut for $(G,\{\{t\} \mid t \in T\})$. 
Further, let $R \subseteq \binom{T}{2}$ be a set of pairs over $T$,
referred to as \emph{cut requests}. A \emph{multicut for $R$ in $G$} is
a set of edges $X \subseteq E$ such that every connected component
in $G-X$ contains at most one member of every pair $\{u,v\}  \in R$.
A \emph{minimum multicut for $R$ in $G$} is a multicut for $R$ in $G$ of
minimum cardinality. Similarly, a  \emph{minimum multiway cut for $\cT$ in $G$}
is a multiway cut for $\cT$ in $G$ of minimum cardinality.
For the rest of the paper, we let all cuts implicitly be edge cuts,
unless otherwise specified, hence we generally refer simply to
multiway cuts. 

We define a \emph{multicut-mimicking network} for $T$ in $G$ as a
graph $G'=(V',E')$ such that $T \subseteq V'$ and such that for
every set of cut requests $R \subseteq \binom{T}{2}$, the size of a
minimum multicut for $R$ is equal in $G$ and in $G'$.
We observe that this is equivalent to preserving the sizes of
minimum multiway cuts over all partitions of $T$.

\begin{proposition} \label{prop:multi-equals-multiway}
  A graph $G'$ with $T \subseteq V(G')$ is a multicut-mimicking
  network for $T$ in $G$ if and only if, for every partition $\cT$
  of $T$, the size of a minimum multiway cut for $\cT$
  is equal in $G$ and in $G'$.
\end{proposition}
\begin{proof}
  It is clear that the condition is necessary, since for any partition
  $\cT$ of $T$ we could form the set $R$ of all pairs over $T$ which
  lie in distinct parts of $\cT$, and a multicut for $R$ is then necessarily
  a multiway cut for $\cT$.  To see that the condition is also sufficient,
  consider an arbitrary set of cut requests $R \subseteq \binom{T}{2}$
  and let $X$ be a minimum multicut for $(G,R)$. Let $\cT$ be the
  partition of $T$ in $G-X$ according to connected components.
  Then $X$ is a multiway cut for $\cT$, and any multiway cut for $\cT$
  is also a multicut for $R$. Hence the size of a minimum multicut for $R$
  is precisely the size of a minimum multiway cut for $\cT$. 
\end{proof}

As a slightly sharper notion, a \emph{multicut-covering set} for $(G,T)$ is
a set $Z \subseteq E(G)$ such that for every set of cut requests $R \subseteq \binom{T}{2}$,
there is a minimum multicut $X$ for $R$ in $G$ such that $X \subseteq Z$.
Note that a multicut-covering set $Z$ is essentially equivalent to a multicut-mimicking
network formed by contraction (contracting all edges of $E(G) \setminus Z$). 
Our main result in this paper is the existence of a multicut-covering
set of size quasipolynomial in $k=\cpc(T)$ in any undirected graph $G$.
Furthermore, such a set can be computed in polynomial time, subject to
the existence of certain approximation algorithms that we will make precise
later in this section.

A multicut-covering set is a generalization of a \emph{cut-covering set},
used in previous work~\cite{KratschW20JACM}.
Formally, a cut-covering set for $(G,T)$ is a set $Z \subseteq E(G)$
such that for any partition $T=A \cup B$ there is an $(A,B)$-min cut
$X$ in $G$ with $X \subseteq Z$. By previous work, if $(G,T)$ is a terminal
network with $\cpc(T)=k$, then a cut-covering set of $O(k^3)$ edges
can be computed in randomized polynomial time~\cite{KratschW20JACM}.
We observe that this gives us a 2-approximate
multicut-covering set.

\begin{lemma} \label{cut-cover-is-apxmcc}
  Let $(G,T)$ be a terminal network and let $Z$ be a cut-covering
  set for edge cuts over $T$.  Then $Z$ is a 2-approximate
  multicut-covering set. 
\end{lemma}
\begin{proof}
  Let $\cT=T_1 \cup \ldots \cup T_s$ be a partition of $T$ and let $X$
  be a minimum multiway cut for $\cT$. For $i \in [s]$,
  let $\lambda_i=\lambda(T_i, T \setminus T_i)$ be the size
  of an isolating min-cut for $T_i$ in $G$.  Then $Z$ contains
  a $(T_i, T \setminus T_i)$-cut of cardinality $\lambda_i$
  for every $i \in [s]$, and by taking their union we get a solution
  $X' \subseteq Z$ with $|X'| \leq \sum_{i=1}^s \lambda_i$.

  It is known~\cite[Cor.~73.2e]{SchrijverBook} that there exists a
  half-integral multiflow in $G$ for $\cT$ of value
  $\frac{1}{2} \sum_{i=1}^s \lambda_i \geq |X'|/2$.
  Hence $|X| \geq |X'|/2$, and $X'$ is a 2-approximate multiway cut
  for $\cT$.  By the argument of Prop.~\ref{prop:multi-equals-multiway},
  $X'$ is also a 2-approximate multicut-mimicking network for
  $(G,T)$. 
\end{proof}

Thus, in randomized polynomial time we can compute a 2-approximate multicut-covering
set for a terminal network $(G,T)$ with $O(\cpc(T)^3)$ edges.

\subsection{Graph separation algorithms}
\label{sec:graphsep}

The central technical approximation assumption needed in this paper is
the following. For a graph $G$ with a set of terminals $T$, define
the \emph{$T$-capacity of $S$ in $G$} as
\[
  \cpc_T(S) =  \cpc_G(T\cap S) + \delta_G(S).
\]
Then we define the following notion.\footnote{For the main results of the paper, 
  it suffices to assume a specialised version that focuses on cuts
  that cut through the terminal set.  The details were worked out 
  in the preliminary version of this paper~\cite{selfICALP},
  under the name \emph{sublogarithmic terminal expansion tester}.
  We use here a simplified definition that suffices for our results.}

\begin{definition}[Quasipolynomial expansion tester]
  \label{def:set-finder}
  Let $(G,T)$ be a terminal network with $\cpc_G(T)=k$.
  A \emph{quasipolynomial expansion tester (with approximation ratio $\alpha$)} is
  a (possibly randomized) algorithm that, given as input $(G,T)$ and
  an integer $c \in \N$, with $c=\Omega(\log k)$, does one of the following.  
  \begin{enumerate}
  \item Either returns a set $S \subset V$ such that $N[S] \neq V(G)$
    and $\cpc_T(S) < |S|^{1/c}$,
  \item or guarantees that for every set $S$ with
    % $\emptyset \subset (S \cap T) \subset T$ and
    $0 < |S| \leq |V(G)|/2$ and $N[S] \neq V(G)$ we have $\cpc_T(S) \geq |S|^{1/c}/\alpha$.
    % REMARK: intermediate requirement: \geq \delta(S \setminus T) \geq
  \end{enumerate}
  More generally, we allow $\alpha \colon \N \to \N$ to be a function
  depending on $|S|$ in addition to $n$ and $k$,
  in which case the guarantee in the second item is $\cpc_T(S) \geq |S|^{1/c}/\alpha(|S|)$.   
  We say that $(G,T)$ is \emph{$(\alpha,c)$-dense} if
  case 2 above applies, i.e., for every set $S$ with $0 < |S| \leq |V(G)|/2$ and $N[S] \neq V(G)$
  we have $\cpc_T(S) \geq |S|^{1/c}/\alpha(|S|)$.
\end{definition}

%Note that the algorithm is allowed to return a set $S$ which does not
%intersect $T$ (or which contains all of $T$), assuming that it is
%sufficiently non-expanding.

We note that such algorithms follow from 
approximation algorithms for \textsc{Small Set Expansion};
indeed, the problem definitions are almost identical, except for the
parameter $c$.
Let $G=(V,E)$ be a graph and $S \subseteq V$ a set of vertices.
The \emph{edge expansion} of $S$ is
\[
\Phi(S) := \frac{\delta(S)}{|S|}.
\]
For a real number $\rho \in (0, 1/2]$, one also defines
the \emph{small set expansion}
\[
\Phi_\rho(G) := \min_{S \subseteq V, |S| \leq \rho n} \Phi(S).
\]
In particular, for a value $s \in [n/2]$,
$\Phi_{s/n}(G)$ denotes the worst (i.e., minimum) expansion among
subsets of $G$ of size at most $s$. 
Approximation algorithms for \textsc{Small Set Expansion}
imply quasipolynomial expansion testers, as follows.

\begin{lemma} \label{lemma:gives-approx}
  Assume that \textsc{Small Set Expansion} has a bicriteria
  approximation algorithm that on input $(G,\rho)$
  returns a set $S$ with $|S| \leq \beta \rho n$
  and $\Phi(S) \leq \alpha(\rho n) \cdot \Phi_\rho$,
  for some $\alpha, \beta \geq 1$ which may depend on $n$,
  $k=\Phi_\rho \rho n$, and $\rho$. 
  Also assume $c=\Omega(\log k)$ and $\alpha \leq O(\log n)$.
  Then there is a quasipolynomial expansion tester with approximation
  function $\alpha'(|S|) = \Theta(\alpha(|S|) \beta)$.
  %If $\alpha \beta = O(\log^{1-\varepsilon} n \log^{O(1)} (n \cdot \Phi_\rho))$,
  %for some $\varepsilon > 0$, then there is a sublogarithmic terminal
  %expansion tester with ratio $\Theta(\alpha \beta)$ (with $n\cdot \Phi_\rho$ replaced by $k$).
\end{lemma}
\begin{proof}
  Let $\alpha'(s)=2\alpha(s)\beta$.
  Assume that $(G,T)$ is not $(\alpha',c)$-dense for some parameter $c$,
  and let $S \subset V$ be a set witnessing this,
  i.e., $0 < |S| \leq |V(G)|/2$, $N[S] \neq V(G)$,
  and $\cpc_T(S) < |S|^{1/c}/\alpha'(|S|)$. For shorthand write $\alpha'=\alpha'(|S|)$.
  We argue that the set $S \setminus T$ is also a legal return value
  for the algorithm. Note
  \[
    \cpc_T(S \setminus T)=\delta(S \setminus T)
    \leq \delta(S)+\cpc_G(T \cap S) = \cpc_T(S).
  \]
  We also have $|S| > (\alpha' \cpc_T(S))^c \geq (\alpha' \cpc_T(S \setminus T)^c$. 
  Now, recall that \textsc{Minimum Bisection} is FPT parameterized by
  the solution value (i.e., the number of edges cut by an optimal
  solution), with the fastest FPT algorithm running in time
  $O^*(2^{O(p \log p)})$ for parameter $p$~\cite{CyganKLPPSW20plus}.
  Hence we can in polynomial time check for a bisection with
  $p=O(\log n/\log \log n)$ edges, and by replacing a vertex with a
  suitably large clique we can also check for a set $S'$ of cardinality $s$
  with $\delta(S') \leq p$. Hence in the remaining case we assume
  $\cpc_T(S) \geq \delta(S) \geq \Omega(\log n/\log \log n)$.
  Furthermore, by assumption $c=\Omega(\log k)$.
  Hence
  \[
    |S| \geq (\alpha' \cpc_T(S))^c \geq (\log n/\log \log n)^{\log k}=k^{\Omega(\log \log n)},
  \]
  and the difference in size between $|S|$ and $|S \setminus T| \geq |S|-k$
  is negligible.  Hence
  \[
    \Phi(S \setminus T) = \frac{\delta(S \setminus T)}{|S \setminus T|}
    \leq \frac{\cpc_T(S)}{(1-o(1))|S|}
    < (1+o(1))(1/\alpha') |S|^{1/c-1}.
  \]
  Now attach a large clique to every terminal in $G$, say of size
  $\beta |S| + 1$, forming a graph $G'$,
  and call an approximation algorithm for \textsc{Small Set Expansion}
  with a parameter of $\rho = |S|/|V(G')|$.  Assume that the algorithm
  returns a set $S' \subseteq V(G')$. Then $|S'| \leq \beta |S|$, 
  hence $S' \cap T = \emptyset$
  and $\cpc_T(S') = \delta(S')$.  Furthermore % $|S'| \leq \beta |S|$ and
  $\Phi(S') \leq \alpha'' \Phi(S \setminus T)$
  where $\alpha''$ is the approximation guarantee of the algorithm
  on input $(G',\rho)$. Here, the only relevant difference between $\alpha''$
  and $\alpha$ lies in the difference between $|V(G)|=n$ and
  $|V(G'| \leq n + k (\beta |S| + 1)  \leq O(k \beta n)$.
  But since $\alpha$ by assumption depends on $n$ as $O(\log n)$ or slower,
  this difference is a lower-order term. Then
  \[
    \cpc_T(S')=  \delta(S') = |S'| \Phi(S') \leq |S'| \alpha \Phi(S \setminus T)
    < \frac{1+o(1)}{2\beta} |S'|\cdot  |S|^{1/c-1}.
    %< (|S'|/|S|) (1/(2\beta)) (1+o(1)) |S|^{1/c}
    = \frac{1+o(1)}{2\beta} (|S'|/|S|)^{1-1/c}  |S'|^{1/c}.
  \]
  Now, since $|S'|/|S| \leq \beta$ we get $\cpc_T(S') < (1/2)(1+o(1))|S'|^{1/c}$,
  and $S'$ is a valid return value.  By repeating the
  above for all target sizes $|S|=\rho |V(G')|$ from $1$ to $|V(G)|$,
  we can be sure to identify such a set $S'$ if one exists. 
\end{proof}

The strongest general-case approximation is an $O(\log n)$-approximation due to R\"acke~\cite{Racke08decomp},
which is not a useful bound for us; but the 
bicriteria algorithm of Bansal et al.~\cite{BansalFKMNNS14SICOMP}
achieves a ratio of $O(\sqrt{\log n \log (1/\rho)})$,
giving $\alpha(|S|)=O(\log n/\sqrt{\log |S|})$,
which is (just barely) strong enough to get some results.
Stronger bounds, e.g., $\alpha=O(\log k)$ or even $\alpha=O(\sqrt{\log k})$,
would give improved results (cf.~Theorem~\ref{thm:intro:main}).

\section{Multicut-covering sets}\label{sec:mc-mimic}

We now present the main result of the paper, namely the existence of
quasipolynomial multicut-mimicking networks for terminal networks $(G,T)$,
and a method for computing them in randomized polynomial time
using an appropriate quasipolynomial expansion tester.

At a high level, the process works through recursive decomposition of
the graph $G$ across very sparse cuts, treating each piece $G[S]$ of the
recursion as a new instance of multicut-covering set computation,
where the edges of $\partial(S)$ are considered as additional terminals.
The process repeatedly finds a single edge $e \in E(G)$ with a guarantee
that for every set of cut requests $R \subseteq \binom{T}{2}$ there is
a minimum multicut $X$ for $R$ in $G$ such that $e \notin X$. 
We may then contract the edge $e$ and repeat the process.  Thus the
end product is a multicut-mimicking network, and the edges that
survive until the end of the process form a multicut-covering set.

In somewhat more detail, the process uses a variant of the
representative sets approach, which was previously used in
the kernel for \textsc{$s$-Multiway Cut}~\cite{KratschW20JACM}.
Refer to an edge $e$ as \emph{essential for $R$}, for some $R \subseteq \binom{T}{2}$,
if every minimum multicut for $R$ in $G$ contains $e$,
and \emph{essential for $(G,T)$} if it is essential for $R$ 
for some $R \subseteq \binom{T}{2}$.
We use a representative sets approach to return a set of at most $k^c$ edges
which is guaranteed to contain every essential edge, \emph{if}
$(G,T)$ is already $(\alpha,c)$-dense,
for an appropriate value $c=\Omega(\alpha \log k)$.
On the other hand, if $(G,T)$ is not $(\alpha,c)$-dense, 
then (by careful choice of parameters)
we can identify a cut through $G$ which is sufficiently sparse
that we can reduce the size of one side of this cut via a recursive call.
This gives a tradeoff between the size of the resulting
multicut-covering set and the denseness-guarantee we may assume
through the approximation algorithm.
When $\alpha$ is constant (or, more precisely, independent of $|S|$)
then this analysis gives a quite simple bound for an 
algorithm that computes a multicut-covering set of $k^c$ edges.

Unfortunately, the best bound on $\alpha$ independent of $|S|$
is just $\alpha=O(\log n)$, in which case the bound $k^c$
with $c=\Omega(\log n)$ is vacuous. We therefore also
perform a more careful analysis using the \textsc{Small Set Expansion}-
approximation algorithm of Bansal et al.~\cite{BansalFKMNNS14SICOMP},
and show that it allows us to compute a multicut-covering set
of $k^{O(\log^3 k)}$ edges in polynomial time.

\subsection{Recursive replacement}
\label{sec:rec-decomp}

We now present the recursive decomposition step in detail.
Let $(G,T)$ be a terminal network with $\cpc_G(T)=k$.
For a set $S \subseteq V$, we define
the graph
\[
  G_S = G[N_G[S]]-E(N_G(S)),
\]
i.e., $G_S$ equals the graph $G[S]$ with the edges of $\partial(S)$
added back in. We also denote
\[
  T(S) = (T \cap S) \cup N_G(S)
\]
as the \emph{terminals of $S$}.  Under these definitions,
the \emph{$T$-capacity of $S$ in $G$} has two equivalent definitions as
\[
  \cpc_T(S) = \cpc_{G_S}(T(S)) = \cpc_G(T \cap S) + \delta_G(S).
\]
The \emph{recursive instance at $S$} consists of the terminal network $(G_S, T(S))$.
This is the basis of our recursive replacement procedure. 
Indeed, we show the following.  Note that we consider 
 $E(G_S) \subseteq E(G)$ in the following.
 
\begin{lemma} \label{lemma:recursive-may-contract}
  Let $(G_S, T(S))$ be the recursive instance at $S$ for some $S \subseteq V(G)$.
  Let $Z_S$ be a multicut-covering
  set for $(G_S, T(S))$ and let $e \in E(G_S) \setminus Z_S$.
  Then $e$ is not essential for $(G,T)$.
\end{lemma}
\begin{proof}
  By Prop.~\ref{prop:multi-equals-multiway}, it is sufficient to
  consider partitions $\cT$ of $T$ and minimum multiway cuts $X$ for $\cT$.
  Let $\cT$ be some partition of $T$, and let $X$ be a minimum
  multiway cut for $\cT$ in $G$. Let $\cT'$ be the partition of $T(S)$
  induced by the connected components of $G-X$ and $X_S=X \cap E(G_S)$.
  Then $X_S$ is a multiway cut for $\cT'$ in $G_S$.  Indeed, any path $P$
  in $G_S-X_S$ between distinct parts of $\cT'$ also exists in $G-X$.
  If $\cT'$ consists of a single part, then we have $X_S=\emptyset$, 
  as otherwise either $X$ contains an edge $uv$ whose both endpoints
  lie in the same connected component of $G-X$, or $G-X$ contains
  a connected component with no terminals, both of which contradict
  that $X$ is of minimum cardinality. Otherwise, by assumption there
  is a minimum multiway cut $X_S'$ for $\cT'$ in $G_S$ such that $e \notin X'$.
  We claim that $X':=(X \setminus X_S) \cup X_S'$ is a minimum
  multiway cut for $\cT$ in $G$.  Note that $|X'| \leq |X|$,
  hence it remains to show that $X'$ is a multiway cut.
  Assume for a contradiction that $G-X'$ contains a path $P$
  connecting different parts of $\cT$, and consider the partition of
  $P$ into subpaths induced by splitting at every vertex of $T(S)$
  that $P$ intersects. Note that every such subpath is either
  contained in $E(G_S)$ or disjoint from $E(G_S)$, and by assumption
  at least one such subpath is contained in $E(G_S)$, as otherwise $P$ uses only edges also
  present in $G-X$. But every such subpath goes between two vertices
  of $T(S)$ which lie in the same connected component of $G-X$ by
  definition of $\cT'$. Thus every such subpath starts and ends in a
  single connected component of $G-X$, contradicting that $P$ starts
  and ends in different components.   
  Therefore $X'$ is a minimum multiway cut for $\cT$ in $G$.
  Since $e \notin X'$ we are done.
\end{proof}

Let us also briefly note the formal correctness of contracting a
non-essential edge. Let $G/e$ denote the result of contracting $e$ in $G$.

\begin{proposition} \label{prop:contract}
  Let $e \in E(G)$ be a non-essential edge. Then for every $X \subseteq E(G)$
  with $e \notin X$, and every partition $\cT$ of~$T$, $X$ is
  a multiway cut for $\cT$ in $G$ if and only if it is a multiway cut
  for $\cT$ in $G/e$. Furthermore, $G/e$ is a multicut-mimicking
  network for $(G,T)$, and any multicut-covering set $Z \subseteq E(G/e)$
  for $(G/e,T)$ is also multicut-covering for $(G,T)$.
\end{proposition}
\begin{proof}
  The first part is clear, since the contraction of an edge in $G-X$ does not change
  the structure of the connected components. Since $e$ is non-essential,
  by assumption there exists such an optimal $X$ with $e \notin X$
  for every partition $\cT$, hence $(G/e,T)$ is a multicut-mimicking
  network. It also follows that an optimal solution for $G$
  always exists in $E(G/e)$, hence a solution-covering set
  for $(G/e,T)$ is also solution-covering for $(G,T)$.
\end{proof}

The process now works as follows.  Recall that $(G,T)$ is
$(\alpha,c)$-dense if $\cpc_T(S) \geq |S|^{1/c}/\alpha$
for every set $S$ with $S \cap T \neq \emptyset$ and $|S| \leq |V|/2$.
The main technical result is a marking process that marks all
essential edges for $(G,T)$ on the condition that $(G,T)$ is
$(\alpha,c)$-dense, and which marks at most $k^c$ edges in total.
In such a case, we are clearly allowed to select and contract
any unmarked edge of $G$.
Now, assume that $(G,T)$ is not $(\alpha,c)$-dense.  Then by definition
there exists a set $S \subset V$ such that $\cpc_T(S) < |S|^{1/c}/\alpha$.
If we can detect a set $S$ such that $\cpc_T(S) < |S|^{1/c}$,
then we can recursively compute a multicut-covering set $Z_S$ for $(G_S, T(S))$,
consisting of at most $\cpc_T(S)^c < |S|$ edges.  By the above, we may
again select any single edge $e \in E(G_S) \setminus Z_S$ and contract $e$ in $G$.  
In either case, we replace $G$ by a strictly smaller graph
until $|E(G)| \leq k^c$, at which point we are done.

The two ingredients in the above are thus the marking process for
$(\alpha,c)$-dense graphs, which we present next, and the ability to
distinguish the two cases, which has been formalized in the notion of
a quasipolynomial expansion tester.

\subsection{The dense case}
\label{sec:dense}

Let us now focus on the marking procedure.  Let a terminal network $(G,T)$ 
with $\cpc_G(T)=k$ and an integer $c$ be given.
We show a process that marks essential edges, on the condition that
$(G,T)$ is $(\alpha,c)$-dense, where we initially assume that $\alpha$ is constant.
That is, we prove the following result. The proof takes up the rest of the subsection.

\begin{lemma} \label{lemma:main-dense}
  Assume that $(G,T)$ is $(\alpha,c)$-dense for some constant $\alpha$.
  There is a function $c=\Theta(\alpha \log k)$ and a randomized
  polynomial-time procedure that returns a set of edges $Z \subseteq E(G)$
  such that $|Z| \leq k^c$ and if $(G,T)$ is $(\alpha,c)$-dense 
  then every essential edge for $(G,T)$ is contained in $Z$.
\end{lemma}

In the preliminary version of this paper~\cite{selfICALP}, we gave a
multi-phase marking procedure for this purpose, with a relatively
complex correctness proof.  In this paper, we give a simplified proof,
based around the following observation.

\begin{proposition} \label{prop:eventually-small}
  Let $\cT$ be a partition of $T$ and $X$ a minimum multiway cut for $\cT$.
  Let $V=V_1 \cup \ldots \cup V_s$ be the partition of $V$
  according to the connected components of $G-X$,
  ordered so that $\cpc_T(V_1) \geq \ldots \geq \cpc_T(V_s)$.
  Then for any $i \in [s]$, $\cpc_T(V_i) \leq 3k/i$. 
\end{proposition}
\begin{proof}
  Since every edge of $X$ is incident with at most two components of
  $G-X$, we have $\sum_{i=1}^s \delta(V_i) \leq 2k$.
  The additional contribution to $\cpc_T(V_i)$ from terminals of $T$
  is precisely $k$ in total. Hence $\sum_i \cpc_T(V_i)\leq 3k$. 
  On the other hand, if $\cpc_T(V_i)>3k/i$
  then $\sum_{j=1}^i \cpc_T(V_j) > i \cdot (3k/i)=3k$. 
\end{proof}

Since $(G,T)$ is by assumption $(\alpha,c)$-dense,
it follows that for $i>1$ we have $|V_i| \leq (\alpha 3k/i)^c$.  
If $|V(G)|>k^c$, and if $c$ is large
enough then it follows that almost all vertices of $G$ are found in
the first few components.  We shall see that this suffices to allow
for a simple marking procedure to capture all essential edges of $(G,T)$.

\subsubsection{Matroid constructions}

Before we show the marking procedure, we need some additional
preliminaries.  We refer to Oxley~\cite{OxleyBook2} for more
background on matroids, and to Marx~\cite{Marx09-matroid} for a more
concise, technical presentation, including the presentation of the
representative sets lemma.  For further examples of kernelization
usage of representative sets, see Kratsch and Wahlstr\"om~\cite{KratschW20JACM}.

A \emph{matroid} is a pair $M=(E, \cI)$
where $\cI \subseteq 2^E$ is the \emph{independent sets} of $M$,
subject to the following axioms.
\begin{enumerate}
\item $\emptyset \in \cI$;
\item if $B \in \cI$ and $A \subseteq B$ then $A \in \cI$; and
\item  if $A, B \in \cI$ with $|B|>|A|$ then there exists an element
  $x \in B \setminus A$ such that $A+x \in \cI$.
\end{enumerate}
A \emph{basis} of $M$ is a maximum independent set of $M$;
the \emph{rank} of $M$ is the size of a basis.

Let $A$ be a matrix, and let $E$ label the columns of $A$.
The \emph{column matroid} of $A$ is the matroid $M=(E, \cI)$ where
$S \in \cI$ for $S \subseteq E$ if and only if the 
columns indexed by $S$ are linearly independent.
A matrix $A$ \emph{represents} a matroid $M$ if $M$ is isomorphic to
the column matroid of $A$. We refer to $A$ as a
\emph{linear representation} of $M$. 

We need three classes of matroids to build from. First, for a set $E$,
the \emph{uniform matroid} over $E$ of rank $r$ is the matroid
\[
U(E,r) := (E,\{S \subseteq E \mid |S| \leq r\}).
\]
Uniform matroids are representable over any sufficiently large field.

The second class is a \emph{truncated graphic matroid}.  Given a graph
$G=(E,V)$, the \emph{graphic matroid} of $G$ is the matroid $M(G)=(E,\cI)$
where a set $F \subseteq E$ is independent if and only if $F$ is the
edge set of a forest in $G$. Graphic matroids can be deterministically
represented over all fields. The \emph{$r$-truncation} of a matroid $M=(E,\cI)$
for some $r \in \N$ is the matroid $M'=(E,\cI')$ where $S \in \cI'$ if and
only if $S \in \cI$ and $|S| \leq r$.  Given a linear representation of $M$,
over some field $\F$, a truncation of $M$ can be computed in
randomized polynomial time, possibly by moving to an extension field
of $\F$~\cite{Marx09-matroid}. There are also methods for doing this
deterministically~\cite{LokshtanovMPS18TALG}, but the basic randomized
form will suffice for us.

The final class is more involved.  Let $D=(V,A)$ be a directed graph
and $S \subseteq V$ a set of source vertices. A set $T \subseteq V$
is \emph{linked to $S$ in $D$} if there are $|T|$ pairwise vertex-disjoint
paths starting in $S$ and ending in $T$. Let $U \subseteq V$.
Then
\[
  M(D,S,U)=(U,\{T \subseteq U \mid T \text{ is linked to $S$ in $D$}\})
\]
defines a matroid over $U$, referred to as a \emph{gammoid}. 
Note that by Menger's theorem, a set $T$ is dependent in $M$ if and
only if there is an $(S,T)$-vertex cut in $D$ of cardinality less than
$|T|$ (where the cut is allowed to overlap $S$ and $T$). 
Like uniform matroids, gammoids are representable over any
sufficiently large field, and a representation can be computed in
randomized polynomial time~\cite{OxleyBook2,Marx09-matroid}.
We will work over a variant of gammoids we refer to as 
\emph{edge-cut gammoids}, which are defined as gammoids, except in
terms of edge cuts instead of vertex cuts.  Informally, for a graph $G=(V,E)$
and a set of source vertices $S \subseteq V$, the edge-cut gammoid of $(G,S)$
is a matroid on a ground set of edges, where a set $F$ of edges is independent
if and only if it can be linked to $S$ via pairwise edge-disjoint paths. 
However, we also need to introduce the ``edge version'' of
\emph{sink-only copies} of vertices, as used in previous
work~\cite{KratschW20JACM}.  That is, we introduce a second set
$E'=\{e' \mid e \in E\}$ containing copies of edges $e \in E$
which can only be used as the endpoints of linkages, not as initial or
intermediate edges. 

More formally, for a graph $G=(V,E)$ and a set of source vertices
$S \subseteq E$ we perform the following transformation.
\begin{enumerate}
\item Let $L(G)$ be the line graph of $G$, i.e., the vertices of
  $L(G)$ are $V(L(G))=\{z_e \mid e \in E(G)\}$,
  and $z_ez_f \in E(L(G))$ if and only if $e \cap f \neq \emptyset$.
  Let $S_D \subseteq V(L(G))$ be the vertices of $L(G)$ corresponding to
  the edges $E(S,V)$ in $G$.
\item Convert $L(G)$ to a directed graph $D_G$ by replacing every edge
  $z_ez_f \in E(L_G)$ by a pair of directed edges $(z_e,z_f)$, $(z_f,z_e)$ in $E(D_G)$.
\item Finally, for every vertex $z_e \in V(D_G)$ introduce a new
  vertex $z_e'$, and create a directed edge $(v,z_e')$ for every
  edge $(v,z_e)$ in $D_G$. 
\end{enumerate}
Slightly abusing notation, we let $E$ refer to the vertices $z_e$ in $D_G$,
and we let $E'$ refer to the vertices $z_e'$ in $D_G$. The
\emph{edge-cut gammoid of $(G,S)$} is the gammoid
$(D_G, S_D, E \cup E')$.
Let us observe the resulting notion of independence.

\begin{proposition}
  Let $G=(V,E)$ and $S \subseteq V$ be given.  Let $M=(E \cup E', \cI)$
  be the edge-cut gammoid of $(G,S)$.  Let $X \subseteq E \cup E'$ be given,
  and let $F=(X \cap E) \cup \{e \mid e' \in F \cap E'\}$.
  Then $X$ is independent in $M$ if and only if there exists a set $\cP$ of $|X|$
  paths linking $X$ to $S$, where paths are pairwise edge-disjoint
  except that if $\{e, e'\} \subseteq X$ for some edge $e$, then two distinct
  paths in $\cP$ end in $e$. 
\end{proposition}

We let $U(E,p)$ denote the
uniform matroid of rank $p$ on ground set $E(G)$, $M_G(p)$ the
$p$-truncated graphic matroid of $G$, and $M(T)$ the edge-cut gammoid of $(G,T)$.

If $M_1=(E_1,\cI_1)$ and $M_2=(E_2,\cI_2)$ are two matroids
with $E_1 \cap E_2 = \emptyset$, then their
\emph{disjoint union} is the matroid 
\[
M_1 \uplus M_2 = (E_1 \cup E_2, \{I_1 \cup I_2 \mid I_1 \in \cI_1, I_2 \in \cI_2\}).
\]
If $M_1$ and $M_2$ are represented by matrices $A_1$ and $A_2$ over
the same field, then $M_1 \uplus M_2$ is represented by
the matrix
\[
A=
\begin{pmatrix}
  A_1 & 0 \\
  0   & A_2 \\
\end{pmatrix}
\]
We will define matroids $M$ as the disjoint union over several copies
of the base matroids $M(T)$, $M_G(p)$ and $U(E,p)$ defined above.  In such a
case, we refer to the individual base matroids making up $M$ as the
\emph{layers} of $M$. 

\paragraph{Representative sets.}
Our main technical tool is the representative sets lemma,
due to Lov\'asz~\cite{Lovasz1977} and Marx~\cite{Marx09-matroid}.
This result has been important in FPT algorithms~\cite{Marx09-matroid,FominLPS16JACM}
and has been central to the previous kernelization algorithms for cut
problems, including variants of \textsc{Multiway Cut}~\cite{KratschW20JACM}.
We also introduce some further notions. 

\begin{definition}
  Let $M=(E,\cI)$ be a matroid and $X, Y \in \cI$.
  We say that \emph{$Y$ extends $X$ in $M$} if
  $r(X \cup Y)=|X|+|Y|$, or equivalently,
  if $X \cap Y= \emptyset$ and $X \cup Y \in \cI$.
  Furthermore, let $c=O(1)$ be a constant
  and let $\cY \subseteq \binom{E}{c}$.
  We say that a set $\hat \cY \subseteq \cY$
  \emph{represents $\cY$ in $M$} if the following holds:
  For every $X \in \cI$ for which there exists some $Y \in \cY$
  such that $Y$ extends $X$ in $M$, then there exists some $Y' \in \hat \cY$
  such that $Y'$ extends $X$ in $M$. 
\end{definition}

The representative sets lemma now says the following.

\begin{lemma}[representative sets lemma~\cite{Lovasz1977,Marx09-matroid}]
  \label{lemma:repset}
  Let $M=(E, \cI)$ be a linear matroid represented by a matrix $A$ of rank $r+s$, 
  and let $\cY \subseteq \binom{E}{s}$ be a collection of independent sets of $M$,
  where $s=O(1)$. In time polynomial in the size of $A$ and the size of $\cY$, 
  we can compute a set $\hat \cY\subseteq \cY$ of size at most $\binom{r+s}{s}$ 
  which represents $\cY$ in $M$. 
\end{lemma}

We use the following \emph{product form} of the representative
sets lemma, with stronger specialized bounds.
Assume that the rank of $M$ is $r=r_1+\ldots+r_c$, where $r_i$
is the rank of layer $i$ of $M$.
Then Lemma~\ref{lemma:repset} gives a bound on $|\hat \cY|$ as 
$\Theta((r_1+\ldots+r_c)^c)$, but the following bound is significantly
better when the layers of $M$ have different rank. 

\begin{lemma}[{\cite[Lemma~3.4]{KratschW20JACM}}]
  \label{lemma:repset-product}
  Let $M=(E,\cI)$ be a linear matroid, given as the disjoint union of
  $c$ matroids $M_i=(E_i,\cI_i)$, where $M_i$ has rank $r_i$.
  Let $\cY \subseteq \binom{E}{c}$ be such that every set $Y \in \cY$
  contains precisely one member in each layer $M_i$ of $M$.
  Then the representative set $\hat \cY \subseteq \cY$ computed by the
  representative sets lemma will have  $|\hat \cY| \leq \prod_{i=1}^c r_i$.
\end{lemma}

\subsubsection{The marking step}

For the marking process, fix an integer $i_0$ (to be specified later).
We define a matroid $M$ as a function of $i_0$ as the disjoint union
of $i_0-1$ copies of the edge-cut gammoid $M(T)$ on disjoint copies
of the ground set, one copy of $M_G(k^{c-i_0})$, and one copy of $U(E,k)$.
That is,
$M=M_1 \uplus \cdots \uplus M_{i_0+1}$, where $M_1$ through $M_{i_0-1}$ are copies
of $M(T)$ on disjoint copies of the ground set, $M_{i_0}=M_G(k^{c-i_0})$, and $M_{i_0+1}=U(E,k)$.
We refer to the first $i_0-1$ layers in $M$ as the \emph{gammoid layers}
and the last two as the \emph{graphic matroid layer} and the \emph{uniform layer}. 
Note that a linear representation of $M$ over some common field $\F$
can be computed in randomized polynomial time, since every layer of
$M$ can be represented over any sufficiently large field.

For each edge $e \in E$, let $t(e)$ be the set that contains
a copy of $z_e'$ in every gammoid layer, and a copy of $e$ in the
graphic matroid and uniform layers. Let
\[
F = \{t(e) \mid e \in E\}.
\]
We compute a representative set $\hat F \subseteq F$ in the matroid $M$,
and let  $Z \subseteq E$ be the set of edges represented in $\hat F$.
An edge $e \in E$ is \emph{marked} if $e \in Z$.
We finish the description by observing the bound on the number of
marked edges. 

\begin{lemma}
  The total number of marked edges is at most $k^c$.
\end{lemma}
\begin{proof}
  Follows directly from the product form of the representative sets lemma.
\end{proof}

Finally, we note the correctness condition for the marking.
Consider a partition $\cT$ of $T$ and a corresponding minimum multiway
cut $X \subseteq E$. Note that $|X| \leq k$ since $E(T,V)$ is a
multiway cut for every partition, and say that $X$ is \emph{covered}
if all edges essential for $\cT$ are marked.  We then have the following. 

\begin{lemma}
  \label{lemma:updated-criterion}
  Let $V=V_1 \cup \ldots \cup V_s$ be the partition of $G-X$ into
  connected components, where $|V_1| \geq \ldots \geq |V_s|$. 
  If $|\bigcup_{i=i_0}^s V_i| \leq k^{c-i_0}$,
  then $X$ is covered.
\end{lemma}
\begin{proof}
  Let $e \in X$ be an edge which is essential for $\cT$.
  Let $\cT=\{T_1,\ldots,T_s\}$ where $T_i=T \cap V_i$ for $i \in [s]$.
  Finally, define an independent set $I$ in $M$ as follows.  
  In the $i$:th gammoid layer, $i < i_0$, $I$ contains copies of
  vertices $z_e$ from the edges of $\partial(T_i) \cup X$. 
  In the graphic matroid layer, $I$ contains a spanning forest 
  for components $V_{i_0}$ through $V_s$. In the final layer,
  $I$ contains the edges of $X-e$. We claim that $t(f)$
  extends $I$ if and only if $f=e$. 

  For the easier direction, we note that $t(f)$ cannot extend $I$ if $e \neq f$. 
  If $f \in E(V_i)$ for some $i < i_0$, then $f$ fails to extend $I$ in layer $i$.
  If $f \in E(V_i)$ for $i \geq i_0$, then $f$ fails to extend $I$ in the graphic 
  matroid layer. Finally, if $f \in X$ then $f$ fails to extend $I$ in the 
  uniform matroid layer. Hence it remains to show that $t(e)$ extends $I$. 

  For the gammoid layers, this works precisely as in~\cite{KratschW20JACM}.
  As noted in~\cite{KratschW20JACM} (Prop.~1), whether a sink-only copy $v'$
  extends a set $U$ in a gammoid $(D,S)$ depends on whether the original copy $v$ is
  contained in the $(S,U)$-min cut \emph{closest to $S$}.
  Here, including $\partial(T_i)$ in $I$ in layer $i$ effectively
  turns this condition into a cut between $X$ and $\delta(T \setminus T_i)$.
  Hence if $e'$ does not extend $X \cup \partial(T_i)$, then
  there is a min-cut $X_2$ between $X$ and $T \setminus T_i$
  that is closer to $T \setminus T_i$ than $X$, 
  and $e \notin X_2$.  This contradicts that $e$ is essential for $\cT$.

  For the last two layers, the statement is trivial.
  Hence $t(f)$ extends $I$ if and only if $f=e$, as promised, and $e \in Z$.
\end{proof}

\subsubsection{Correctness}

We finish this section by showing that a choice of $c=\Theta(\alpha \log k)$
suffices, when $\alpha$ is constant.  

\begin{lemma} \label{lemma:existence-size}
  Set $i_0=4\alpha$.
  There is a value $c=\Theta(\alpha \log k)$ such that the following holds:
  If $(G,T)$ is $(\alpha,c)$-dense,  then $Z$ contains all essential edges.
\end{lemma}
\begin{proof}
  Let $\cT$ be a partition of $T$ and let $X$ be a minimum multiway
  cut for $\cT$.  Let $V=V_1 \cup \ldots \cup V_s$ be the connected
  components of $G-X$ sorted by decreasing value of $\cpc_T(V_i)$. 
  By Prop.~\ref{prop:eventually-small}, $\cpc_T(V_i) \leq 3k/i_0$
  for $i \geq i_0$; hence $|V_i| \leq (\alpha 3k/i_0)^c = (3k/4)^c$ by the denseness
  guarantee on $(G,T)$. Thus
  \[
    \sum_{i=i_0}^s |V_i| \leq s \cdot (3k/4)^c \leq k \cdot k^c \cdot (3/4)^c.
  \]
  For $c = \log_{4/3}(k^{i_0+1})=\Theta(\alpha \log k)$, this is upper-bounded
  by $k^{c-i_0}$, hence $X$ is covered by Lemma~\ref{lemma:updated-criterion}.
  Hence $Z$ contains all essential edges. 
\end{proof}

In Theorem~\ref{thm:main}, we combine this with the decomposition
described previously to show the existence of a multicut-covering set
with $k^{O(\log k)}$ edges.

\subsection{A constructive result} \label{sec:constructive}

Now, with some slightly more involved calculations and a higher degree
in the exponent, we show that this can be achieved constructively.
The marking process is the same as in the last section, but we need to
be more careful with the constants.

For this section, let $c=\lfloor (\log n)/(\log k)\rfloor-1$, to
ensure that the number of edges marked is $k^c < n$.
We will treat this as $c=(1-o(1))\log n/\log k$.
We show that such a marking process is possible until we reach a bound
of $c=O(\log^3 k)$. Hence, we assume $|V| = k^{\Omega(\log^3 k)}$ in
the sequel.

We use the \textsc{Small Set Expansion} approximation algorithm of
Bansal et al.~\cite{BansalFKMNNS14SICOMP}, which has a ratio
of $\alpha(|S|) = O(\log n/\sqrt{\log |S|})$. 
By Lemma~\ref{lemma:gives-approx}
we may then assume that for any non-empty set $S \subseteq V$, 
$|S|  \leq |V|/2$, we have $\cpc_T(S) \geq |S|^{1/c}/\alpha(|S|)$.

By Lemma~\ref{lemma:updated-criterion}, we wish to find a threshold
value $i_0$ such that the total number of vertices in components $V_i$
for $i \geq i_0$ is at most $k^{c-i_0}$, where it suffices to show
that $|V_i| \leq k^{c-i_0-1}$ for every $i \geq i_0$. 
We set $i_0=\Theta(\sqrt{c \log k})$ with a constant factor to be decided.

Let $(G,T)$ be $(\alpha,c)$-dense with $\alpha$ and $c$ as above.
Let $\cT$ be a partition of $T$ and $X$ a minimum multiway cut for $\cT$.
We show that the marking process with parameters $c$ and $i_0$
cover all essential edges of $X$. 

\begin{lemma}
  Let $V=V_1 \cup \ldots \cup V_s$ be the partition corresponding to
  connected components of $G-X$, ordered in decreasing value of $\cpc_T(V_i)$.
  With the above parameters, $|V_i| \leq k^{c-i_0-1}$ for every $i \geq i_0$.
\end{lemma}
\begin{proof}
  Assume for a contradiction that for some $i \geq i_0$,
  $|V_i| > k^{c-i_0-1}$.  By the SSE approximation we use,
  for some $p=O(1)$ we then have
  \[
    \alpha(|V_i|) \leq \frac{p\log n}{\sqrt{(c-i_0-1)\log k}} \leq
    (1+o(1)) \frac{pc\sqrt{\log k}}{\sqrt{c-i_0-1}},
  \]
  where we can absorb the $1+o(1)$ factor into $p$.
  Furthermore fix the constant in $i_0$ so that $i_0 = 6p\sqrt{c \log k}$. 
  By Prop.~\ref{prop:eventually-small} we have
  $\cpc_T(V_i) \leq 3k/i \leq 3k/i_0$, hence the denseness guarantee is
  \[
    |V_i| \leq (3\alpha k/i_0)^c =
    k^c \cdot \left(\frac{3pc\sqrt{\log k}}{6p\sqrt{c \log k}\sqrt{c-i_0-1}}\right)^c =
    k^c \cdot \left(\frac 1 2 \cdot \sqrt{\frac{c}{(c-i_0-1)}}\right)^c.
  \]
  Since $c = \Omega(\log^3 k)$, we have $i_0=\Theta(\sqrt{c \log k})=o(c)$.
  Hence $c/(c-i_0-1)=1+(i_0+1)/(c-i_0-1)=1+o(1)$ and we claim
  \[
    |V_i| \leq k^c \cdot \left( \frac{1+o(1)}{2}\right)^c \leq k^{c-i_0-1},
  \]
  contradicting our assumption. 
  % \[
%    k^c  \cdot \left(\frac{3pc \sqrt{\log k}}{i \sqrt{(c-i-1)}}\right)^c
%    = k^c \cdot \left( \frac{\sqrt{c}}{2\sqrt{(c-i-1)}} \right)^c.
%  \]
  Indeed, $k^{i_0}=2^{\Theta(c^{1/2}\log^{3/2} k)} \leq 2^c$
  for some $c=\Omega(\log^3 k)$.  Thus the contradiction is
  complete and $|V_i| \leq k^{c-i_0-1}$. 
\end{proof}

Since there are at most $k$ components, it follows
that the number of vertices in total in components $V_i$,
$i \geq i_0$, is less than $k^{c-i_0}$, as required. 
Hence we have the following. 

\begin{lemma} \label{lemma:approximate-success}
  If $|V(G)| \geq k^{\Omega(\log^3 k)}$ and $(G,T)$ is $(\alpha,c)$-dense as above,
  then there is a process that marks all essential edges of $(G,T)$
  while leaving at least one edge unmarked. 
\end{lemma}

\subsection{Completing the result}

We now put the pieces together to show the non-constructive and
constructive bounds on the size of a multicut-covering set. 

%By the above, every terminal network $(G,T)$ that is $(\alpha,c)$-dense
%for appropriate $c$ and $\alpha$ has a multicut-covering set of at most $k^c$ edges,
%which can be computed in randomized polynomial time.
%We extend the result to any $(G,T)$, using a
%sublogarithmic terminal expansion tester.

\begin{theorem}[Theorem~\ref{thm:intro:main} restated]
  \label{thm:main}
  Let $A$ be a quasipolynomial expansion tester with ratio $\alpha(n,k)$.
  Let $(G,T)$ be a terminal network with $\cpc_G(T)=k$. There is a multicut-covering
  set $Z \subseteq E(G)$ with $|Z| \leq k^{O(\alpha(n,k) \log k)}$, which furthermore can be
  computed in randomized polynomial time using  calls to $A$. 
\end{theorem}
\begin{proof}
  Set $c=\Theta(\alpha \log k)$ as in Lemma~\ref{lemma:existence-size}. 
  If $|E(G)| \leq k^c$ then return $Z=E(G)$. Otherwise, we compute a
  non-essential edge $e$ as follows.  Call $A$ on $(G,T,c)$. 
  If $A$ reports that $(G,T)$ is $(\alpha,c)$-dense, then Lemma~\ref{lemma:main-dense}
  applies. Compute a set $Z$ containing all essential edges,
  with $|Z| \leq k^c$, guaranteeing that there is a non-essential edge
  $e \in E(G) \setminus Z$.
  
  If $A$ returns a set $S \subseteq V(G)$, let $k_S=\cpc_T(S)$.  Let $(G_S, T(S))$ be
  the recursive instance at $S$, and note that $|V(G_S)|=|N_G[S]|<|V|$
  and $|S| > k_S^c$ by definition of $A$. 
  We may now proceed by induction on $|V|$ and assume that we can
  compute a multicut-covering set $Z_S \subseteq E(G_S)$ of size
  $|Z_S| < k_S^c$. To eliminate a corner case, if there is a vertex
  $v \in V(G_S)$ with $v \notin T(S)$ and $d_{G_S}(v)\leq 2$, then 
  delete $v$ if $v$ is a leaf, otherwise contract one edge
  incident with $v$.  Note that since $v \notin T(S)$ we have
  $d_G(v)=d_{G_S}(v)$ and $v \notin T$, hence these reduction rules
  are clearly correct. If this rule does not apply, there must be some edge
  $e \in E(G_S) \setminus Z_S$, and by construction $e$ corresponds
  directly to an edge in $G$. Again, we have found a non-essential edge.

  By Prop.~\ref{prop:contract} we may now contract $e$ in $G$ and repeat.
  This yields a graph $G'$ with $|V(G')|<|V|$, hence by induction we
  can create a multicut-covering set $Z$ for $G'$, which is also a
  multicut-covering set of $G$ by Prop.~\ref{prop:contract}.
  Hence we can compute a multicut-covering set $Z$ with $|Z| \leq k^c$. 
\end{proof}

We observe the following consequences.

\begin{corollary} \label{corollary:actual}
  Let $(G,T)$ be a terminal network with $\cpc_G(T)=k$. The following holds.
  \begin{enumerate}
  \item There is a multicut-mimicking network for $(G,T)$ with $k^{O(\log k)}$ edges.
  \item A multicut-mimicking network with $k^{O(\log^3 k)}$ edges can
    be computed in randomized polynomial time.
  %\item If there is a sublogarithmic terminal expansion tester -- in particular, if \textsc{Small Set Expansion} has an approximation ratio as in Theorem~\ref{thm:main} -- then a multicut-mimicking network of size quasipolynomial in $k$ can be computed in randomized polynomial time. 
  \end{enumerate}  
\end{corollary}
\begin{proof}
  The first is immediate using $\alpha(n,k)=1$. 
  For the second, we need to use Lemma~\ref{lemma:approximate-success}.
  We assume that $|E(G)| > k^{\Omega(\log^3 k)}$ as required by
  Lemma~\ref{lemma:approximate-success}, or else we return $(G,T)$.
  As in Theorem~\ref{thm:main} it suffices to locate a single
  non-essential edge. By Bansal et al.~\cite{BansalFKMNNS14SICOMP}
  and Lemma~\ref{lemma:gives-approx}, there is a quasipolynomial expansion tester
  with a ratio $\alpha(|S|) = O(\log n/\sqrt{\log |S|})$.
  Call this algorithm with $(G,T,c)$ where $c =(1-o(1)) \log n/\log k$
  as in Section~\ref{sec:constructive}.  If it finds a set $S$,
  recurse on $S$ as in Theorem~\ref{thm:main}; otherwise
  Lemma~\ref{lemma:approximate-success} applies.
  In both cases we locate a non-essential edge.
  Eventually we reach the threshold where $|E(G)|=k^{O(\log^3 k)}$
  and use $Z=E(G)$ as multicut-covering set. We then return the
  resulting terminal network $(G,T)$ as multicut-covering set.
  % For the second, all that remains is to clean up the value $|Z|$.  For this,
  %let $\alpha(n,k) \leq \log^{1-\varepsilon} n \log^d k$ and $c=b \alpha \log k$,
  %for some bounded values $b, d$, and first assume that $|Z| \geq |V(G)|=n$.  Then
  %\begin{align*}
  %  n \leq |Z|& < k^{b \alpha \log k} \Rightarrow \\
  %  %2^{\log n} &< 2^{b \alpha \log^2 k} \Rightarrow \\
  %  \log n &< b \alpha \log^2 k \Rightarrow \\
  %  \log n &< b \log^{1-\varepsilon} n \log^{d+2} k \Rightarrow \\
  %  \log^{\varepsilon} n &< b \log^{d+2} k \Rightarrow \\
  %  \log n &< (b \log^{d+2} k)^{1/\varepsilon},
  %\end{align*}
  %hence $|Z| \leq k^{\log^{O(1)} k}$, as promised. 
  %Otherwise, we contract all edges not present in $Z$ and compute a
  %new multicut-covering set $Z'$ for the new system $(G',T)$.
  %Eventually, this process halts, and at this point we will have a
  %multicut-covering set $Z$ with $|Z| \leq k^{\log^{O(1)} k}$ for some
  %graph $G''$ created by contractions from $G$, and by Prop.~\ref{prop:contract}
  %this set $Z$ is also a multicut-covering set for $(G,T)$. 
\end{proof}

Additionally, if better approximation ratios for \textsc{Small Set
  Expansion} exist, in particular ratios $\alpha=O(\log k)$ or better,
then we can improve $k^{O(\log^3 k)}$ to a constructive bound of size $k^{O(\alpha \log k)}$.

\subsection{Kernelization extensions and consequences}

As noted, we get the following consequences.

\begin{corollary} \label{cor:kernels}
  The following problems have randomized quasipolynomial kernels. 
  \begin{enumerate}
  \item\label{res:1} \textsc{Edge Multiway Cut} parameterized by solution size.
  \item \textsc{Edge Multicut} parameterized by the solution size and the number of cut requests.
  \item \textsc{Group Feedback Edge Set} parameterized by solution size, for any group.
  \item \label{res:penultimate} \textsc{Subset Feedback Edge Set} with undeletable edges, parameterized by
    solution size.
  %\item \label{res:last} \textsc{0-Extension} for integer-weighted graphs, parameterized by solution cost.
  \end{enumerate}
\end{corollary}
%
%Most of these consequences are standard given Cor.~\ref{corollary:actual} 
%and results from the literature, but we leave the last result for a
%separate proof since it needs some special attention.
\begin{proof}%[of results \ref{res:1}--\ref{res:penultimate}]
  For \textsc{Edge Multiway Cut}, let $(G, T, k)$ be an input.  Known reduction rules
  can reduce the instance to so that $\cpc_G(T) \leq 2k$~\cite{KratschW20JACM}.
  From this point, the kernel follows.

  For \textsc{Edge Multicut}, let the input be
  $(G, \{(s_1,t_1) \ldots, (s_r,t_r)\}, p)$ and let $k=p+r$. 
  Create a set of $2r$ vertices $T=\{s_1',\ldots, t_r'\}$
  and $p+1$ subdivided parallel edges between $s_i'$ and $s_i$,
  and between $t_i'$ and $t_i$, for each $i \in [r]$. Let $G'$
  be the new graph.  We claim that $I=(G,\{(s_1,t_1),\ldots,(s_r,t_r)\},p)$
  is a positive instance if and only if $I'=(G', \{(s_1',t_1'), \ldots, (s_r',t_r')\}, p)$ is.
  Indeed, any multicut for $I$ is a multicut for $I'$, and every
  multicut for $I'$ containing at most $p$ edges leaves all new
  terminals $s_i'$, $t_i'$ connected to the old terminals $s_i$, $t_i$
  and is hence a multicut for $I$. Furthermore $\cpc_{G'}(T)=(p+1)r=O(k^2)$. 
  Now it suffices to compute a multicut-covering set $Z$
  for $(G',T)$ and contract all edges in $E(G') \setminus Z$.

  For \textsc{Group Feedback Edge Set} (GFES), we follow the approach of~\cite{KratschW20JACM}.
  The input to GFES is a tuple $(G,\phi,k)$, where $\phi$ is a direction-dependent
  labelling of the edges of $G$ from some multiplicative group $\Gamma$, such that for
  every $uv \in E$, $\phi(uv)=\phi(vu)^{-1}$ (where the inverse is the group inverse).
  The goal is to remove $k$ edges such that in the remaining graph, there is
  an assignment $\lambda \colon V(G) \to \Gamma$ such that for any $uv \in E(G)$
  we have $\lambda(v) = \lambda(u) \cdot \phi(uv)$.
  We will not need any assumptions about how the group elements are represented, other than
  the ability to test whether a product of elements $\phi(e)$
  equals the group identity $1_\Gamma$ or not. Refer to a simple cycle $C$ as \emph{unbalanced}
  if $\prod_{e \in E(C)} \phi(e) \neq 1_\Gamma$, with the product taken in order along $C$.
  We first note that GFES has an $O(\log k)$-approximation. 
  Indeed, GFES reduces easily to \textsc{Group Feedback Vertex Set},
  which in turn is a special case of the meta-problem
  \textsc{Biased Graph Cleaning}~\cite{Wahlstrom17SODA}. 
  Lee and Wahlstr\"om~\cite{LeeW20plus} showed that \textsc{Biased Graph Cleaning}
  admits an $O(\log k)$-approximation, using an oracle for testing whether
  cycles are unbalanced. 
  Let $X_0$ be an approximate solution with $|X_0|=O(k \log k)$. 
  Let $T=V(X_0)$ be the endpoints of $X_0$. By assumption, $G-X_0$
  admits an assignment $\lambda \colon V(G) \to \Gamma$ as above,
  and such an assignment $\lambda$ can be computed by starting with an arbitrary value
  from one vertex of each connected component. We now follow~\cite{KratschW20JACM}
  in \emph{untangling} the group labels, so that every edge except those in $X_0$
  receive the identity label by $\phi$.  As in~\cite{KratschW20JACM}, the solution to GFES now 
  corresponds to a multiway cut for some unknown partition $\cT$ of $T$,
  hence the multicut-mimicking network can be used for kernelization. 
  
  \textsc{Subset Feedback Edge Set} with undeletable edges, parameterized by
  solution size, is covered by the previous case, since it is a special case
  of GFES. Indeed, let $S \subseteq E(G)$ be the special edges. We use labels $\phi$ 
  from the group $Z_2^S$ where every edge is labelled by $\phi$ by identity
  except the edges of $S$, which flip one bit of the group element each.
  It is now easy to see that a cycle is balanced if and only if it contains
  no edge from $S$. It is furthermore easy to see that we can implement
  undeletable edges by creating parallel (subdivided) copies of edges,
  using the same group labels. 
\end{proof}

\emph{Remark.}
  In the preliminary version of this paper~\cite{selfICALP},
  we additionally claimed results for the \textsc{0-Extension}
  problem, building on results of Reidl and Wahlstr\"om~\cite{ReidlW18ICALP}.
  Unfortunately, the proof of this in~\cite{selfICALP} is incorrect,
  and we were unable to fix it.  Concretely, using some terminology
  of~\cite{ReidlW18ICALP}, let $I=((G,T),\mu,\tau)$ be an instance of
  \textsc{0-Extension} for a terminal network $(G,T)$ and a metric $\mu$.
  Assume that $I$ has an optimal solution $\lambda$ with at most $k$
  crossing edges, and that $G$ contains a sparse cut, i.e.,
  a set $S \subseteq V(G)$ such that $\cpc_T(S) < |S|^{1/c} < k$. 
  Then the proof attempt in the conference version~\cite{selfICALP}
  implicitly assumes that $\lambda$ requires only $O(\cpc_T(S))$
  crossing edges inside $G[S]$. This does not appear to hold.
  Therefore we retract our previous claims in~\cite{selfICALP}
  regarding quasipolynomial metric sparsifiers. 
  The results of Reidl and Wahlstr\"om~\cite{ReidlW18ICALP} are not
  affected by this. 

\section{Discussion}

We defined the notion of a \emph{multicut-mimicking network},
and showed that every terminal network $(G,T)$ with $k=\cpc_G(T)$
admits one of size $k^{O(\log k)}$, and that a
multicut-mimicking network of size $k^{O(\log^3 k)}$ can be computed
in randomized polynomial time. 
The mimicking network is constructed via contractions on $G$,
i.e., it simply consists of a set of edges which form a \emph{multicut-covering set}.
As a consequence of such a result, a range of parameterized problems,
starting from \textsc{Edge Multiway Cut}, have randomized quasipolynomial
kernels. 

A first question is how to bridge the gap between $k^{O(\log k)}$ and
$k^{O(\log^3 k)}$.  This may be partially possible,
depending on the precise approximation guarantee available for 
\textsc{Small Set Expansion}; but a complete bridging via this
approach would require a constant-factor approximation for SSE,
which has been conjectured not to exist under the \emph{small set expansion hypothesis}.
A more difficult question is what the correct size of a
multicut-mimicking network is in general, and whether one of
polynomial size exists and can be efficiently computed.
A positive solution would confirm the existence of a polynomial kernel
for \textsc{Edge Multiway Cut}, which is one of the most significant
open questions in kernelization.

We finally note two questions in different directions. First, all
results in this paper relate only to edge deletion problems.
Can an extension of the method be used to prove the existence or efficient computability of a
quasipolynomial mimicking network for \emph{vertex deletion} multicut
behaviour? This may introduce significant additional difficulties,
especially considering that approximation algorithms for sparse
vertex cuts are less well-developed than algorithms for
\textsc{Small Set Expansion}.

Second, the preliminary version of this paper~\cite{selfICALP}
contained mistaken claims about existence of quasipolynomial kernels
and ``metric sparsifiers'' for \textsc{0-Extension} instances, subject
to a bound on the number of crossing edges in a solution.
Can such a result be established? On the other hand, can evidence be
established against the existence of a polynomial-sized exact metric
sparsifier, depending on a bound on the number of crossing edges $k$
of a solution, or is even such a result plausible?

\bibliographystyle{plainurl}
\bibliography{allrefs,local}

\end{document}